\documentclass[11pt,onecolumn,final]{article} 
\usepackage[latin1]{inputenc}
\usepackage{fullpage}
\usepackage{amsfonts}
\usepackage{amsmath}
\usepackage{amssymb}
\usepackage{amsthm}
\usepackage{graphicx}
\usepackage[ruled,longend]{algorithm2e}
\usepackage{stmaryrd}
\usepackage{cite}

\usepackage{versions}
\excludeversion{INUTILE}\excludeversion{TEMPORARY}
\excludeversion{JOHANNES}\excludeversion{JEREMY}\includeversion{LONG}\includeversion{VLONG}
\excludeversion{FORJOURNAL}

\newcommand{\Oh}{\ensuremath{{\cal O}}}
\newcommand{\entropy}{\ensuremath{{\cal H}}}

\DeclareMathOperator*{\polylog}{polylg}

\newcommand{\lrm}{Left-to-Right-Minima}  

\newcommand{\rmq}{\ensuremath{\text{\sc rmq}}}
\newcommand{\lca}{\ensuremath{\text{\sc lca}}}
\newcommand{\rqq}{\ensuremath{\text{\sc rqq}}}
\newcommand{\psv}{\ensuremath{\text{\sc psv}}}
\newcommand{\access}{\ensuremath{\textit{access}}}
\newcommand{\rank}{\ensuremath{\textit{rank}}}
\newcommand{\select}{\ensuremath{\textit{select}}}

\newcommand{\map}{\ensuremath{\textit{map}}}
\newcommand{\unmap}{\ensuremath{\textit{unmap}}}

\newcommand{\nwhatever}{\ensuremath{\mathsf{X}}}        
\newcommand{\nPartition}{\ensuremath{\mathsf{nSeq}}}            
\newcommand{\vPartition}{\ensuremath{\mathsf{Seq}}}                    
\newcommand{\vLRM}{\ensuremath{\mathsf{LRM}}}                          
\newcommand{\runs}{\ensuremath{\mathsf{Runs}}} \def\vRuns{\runs}       
\newcommand{\nruns}{\ensuremath{\mathsf{nRuns}}}\def\nRuns{\nruns}     
\newcommand{\nsruns}{\ensuremath{\mathsf{nSRuns}}}\def\nSRuns{\nsruns} 
\newcommand{\sus}{\ensuremath{\mathsf{SUS}}}\def\vSUS{\sus}            
\newcommand{\nsus}{\ensuremath{\mathsf{nSUS}}}\def\nSUS{\nsus}         
\newcommand{\ssus}{\ensuremath{\mathsf{SSUS}}}\def\vSSUS{\ssus}        
\newcommand{\nssus}{\ensuremath{\mathsf{nSSUS}}}\def\nSSUS{\nssus}     

\newcommand{\nsms}{\ensuremath{\mathsf{nSMS}}}\def\nSMS{\nsms}

\theoremstyle{plain}
\newtheorem{definition}{Definition}
\newtheorem{lemma}[definition]{Lemma}

\newcounter{theoremCounter}\newtheorem{theorem}[theoremCounter]{Theorem}

\theoremstyle{remark}

\title{LRM-Trees: Compressed Indices, \\ Adaptive Sorting, and Compressed Permutations}

\author{
  J{\'e}r{\'e}my Barbay\thanks{Departamento de Ciencias de la
    Computaci{\'o}n (DCC), Universidad de Chile,
    \texttt{jeremy.barbay@dcc.uchile.cl}}
  ~and
  Johannes Fischer\thanks{Computer Science Department, Karlsruhe
    University, \texttt{johannes.fischer@kit.edu}}
}

\date{}

\begin{document}

\maketitle

\begin{abstract}
  LRM-Trees are an elegant way to partition a sequence of values into
  sorted consecutive blocks, and to express the relative position of
  the first element of each block within a previous block. They were
  used to encode ordinal trees and to index integer arrays in order to
  support range minimum queries on them.  We describe how they yield
  many other convenient results in a variety of areas, from data
  structures to algorithms: some compressed succinct indices for range
  minimum queries; a new adaptive sorting algorithm; and a compressed
  succinct data structure for permutations supporting direct and
  indirect application in time all the shortest as the permutation is
  compressible.
  \begin{VLONG} 
    As part of our review preliminary work, we also give an
    overview of the, sometimes redundant, terminology relative to
    succinct data-structures and indices.
  \end{VLONG}
\end{abstract}

\section{Introduction}

Introduced by Fischer~\cite{fischer10optimal} as an indexing
data structure which supports range minimum queries (RMQ) in constant
time and zero access to the main data,
and by Sadakane and
Navarro~\cite{sadakane10fully}
to support navigation operators on ordinal trees,
\emph{\lrm~Trees} (LRM-Trees) are an elegant way to
partition a sequence of values into sorted consecutive blocks, and to
express the relative position of the first element of each block within
a previous block.

We describe in this extended abstract how the use of LRM-Trees and
variants yields many other convenient results in a variety of areas,
from data structures to algorithms:
\begin{enumerate}
\item We define several compressed succinct indices supporting Range
  Minimum Queries (RMQs), which use less space than the $2n+o(n)$ bits
  used by the succinct index proposed by
  Fischer~\cite{fischer10optimal} when the indexed array is partially
  sorted.  
  Note that although a space of $2n$ bits is optimal \emph{in the
    worst case} over all possible permutations of size $n$, this is
  not necessarily optimal on more restricted classes of permutations.
  For example, if $A=[1,2,\dots,n]$, it is possible to support RMQs on
  $A$ without any additional space.
  Although there is a RMQ succinct index that exploits the
  compressibility of $A$ \cite{fischer08practical}, it only takes
  advantage of repetitions in the input and would still use $2n+o(n)$
  bits for the example above.
\item We propose a new  sorting algorithm and its adaptive analysis,
  asymptotically superior to adaptive merge
  sort~\cite{barbay09compressed},
  and superior in practice to Levcopoulos and Petersson's sorting
  algorithms~\cite{sortingShuffledMonotoneSequences}.
\item We design a compressed succinct data structure for permutations,
  which
  uses less space than the previous compressed succinct data structure
  from Barbay and Navarro~\cite{barbay09compressed}, and
  supports the access operator and its inverse in time all the shortest
  as the permutation is compressible, and range minimum queries and
  previous smaller value queries in constant time.
\end{enumerate}

All our results are in the word RAM model, where it is assumed that we can do arithmetic and logical operations on $w$-bit wide words in $O(1)$ time, and $w = \Omega(\lg n)$.
The following section gives examples of results that have been obtained in this natural model; we start by
giving an overview of the, sometimes redundant, concepts on
succinct data structures and succinct indices.

\section{Previous Work and Concepts}

\subsection{On the Various Types of Succinct Data Structures}
\label{sec:succ-data-struct}

\begin{VLONG}
  Some concepts (e.g., succinct indices and systematic data structures)
  on succinct data structures were invented more than once, at similar
  times but with distinct names, which makes their classification more
  complicated than necessary.
  Given that our results cross several areas (namely, compressed
  succinct data structures for permutations and indices supporting
  range minimum queries), which each use distinct names, we aim in
  this section to clarify the potential overlaps of concepts, to the
  extent of our knowledge.
\end{VLONG}

A \emph{Data Structure} $\mathcal{D}$ (e.g., run encoding of
permutations~\cite{barbay09compressed}) specifies how to encode data
from some \emph{Data Type} $\cal T$ (e.g., permutations) so that to
support the operators specified by a given \emph{Abstract Data Type}
$\cal A$ (e.g., direct and inverse applications).
Naturally, a data structure usually requires more space than a simple
encoding scheme of the same data-type, given that it supports
operators in addition to just memorize the data: the amount of
additional space required is called the \emph{redundancy} of the
data structure.

A \emph{Succinct Data Structure}~\cite{jacobson89space} is a data
structure whose redundancy is asymptotically negligible as compared to
the space required to \emph{encode} the data itself, in the worst or uniform
average case over all instances of fixed size $n$ (e.g., a succinct
data structure for bit vectors using $n+o(n)$ bits).
An \emph{Ultra-Succinct Data Structure}~\cite{jansson07ultra} is a
\emph{compressed} data-structure (w.r.t.~a parameter measuring the
\emph{compressibility} of the data) whose redundancy is asymptotically negligible
as compared to the space required to \emph{encode} the data
in the worst case over all instances for
which the size $n$ is fixed (e.g., an ultra-succinct data-structure for
binary strings~\cite{raman07succinct} uses $nH_0+o(n)$ bits, where
$H_0$ is the entropy (information content) of the string).
\begin{JEREMY}
  Is it a problem that we did not define yet the entropy at this point?
\end{JEREMY}
\begin{JOHANNES}
  No, I think that's common knowledge for STACS-reviewers.
  But there is a problem: as you did it, ``ultra-succinct'' was
  defined the same way as ``succinct.'' I changed it.
\end{JOHANNES}
A \emph{Compressed Succinct Data Structure}~\cite{barbay09compressed}
is a compressed data structure whose redundancy is asymptotically
negligible as compared to the space required to \emph{compress} the
data itself in the worst or average case over all instances for which
the size $n$ is fixed (e.g., a compressed succinct data
structure for binary strings uses
$nH_0+o(nH_0)$ bits).
\begin{JEREMY}
  Add a reference to the ARXIV version of the ISAAC paper about
  alphabet partitioning here?
\end{JEREMY}
\begin{JOHANNES}
  To keep it simple, I changed the example to bit-strings.
\end{JOHANNES}

An \emph{Index} is a structure which, given access to some data
structure $\cal D$ supporting a defined abstract data type $\cal A$
(e.g., a data structure for strings supporting the access operator),
extends the set of operators supported in good time to a more general
abstract data type $\cal A'$ (e.g., {\rank} and {\select} operators on
strings).\footnote{The fundamental {\rank} and {\select} operators on
  a bit-vector $B$ are defined as follows: $\rank_1(B,i)$ gives the
  number of 1's in the prefix $B[1,i]$, and $\select_1(B,i)$ gives the
  position of the $i$-th 1 in $B$, reading $B$ from left to right ($1
  \le i \le n$). Operations $\rank_0(B,i)$ and $\select_0(B,i)$ are
  defined analogously for 0-bits.}  By analogy with succinct data
structures, the space used by an index is called \emph{redundancy}.
A \emph{Succinct Index}~\cite{barbay07succinct2} or \emph{Systematic
  Data Structure}~\cite{gal07cell} $\cal I$ is simply an index whose
redundancy is negligible in comparison to the space required by $\cal
D$ in the worst case over instances of fixed size $n$.
\begin{VLONG}
  The separation between a data structure and its index was
  implicitly used before its
  formalization~\cite{sadakane06squeezing}
  and explicitly to prove lower bounds on the trade-off between space
  and supporting time of succinct data
  structures~\cite{golynski07optimal}.
\end{VLONG}
Of course, if $\cal D$ is a succinct data structure, then the
data structure formed by the union of $\cal D$ and $\cal I$ is a
succinct data structure as well: this modularity permits the
combination of succinct indices for distinct abstract data types on
similar data types~\cite{barbay07succinct2}.
A \emph{Compressed Succinct Index} is an index whose redundancy is
negligible in comparison to the space required by $\cal D$ in the
worst case over instances of fixed size $n$, as well as decreasing
with a given measure of compressibility of the index (e.g. the
short-cut data-structure~\cite{munro03succinct} supporting
$\pi^{-1}()$ uses space inversely proportional to the length of cycles
in the permutation $\pi$).
\begin{JEREMY}
  We had forgotten to define the term ``Compressed Succinct Index'',
  when this is one of our own results!!! Once written like this, it
  looks a bit arbitrary: the measure of compressibility of the index
  is not necessarily the measure of compressibility of the data...
\end{JEREMY}
\begin{JOHANNES}
  And now we have a *different* terminology for data structures and indices:
  Our first result with $2\nruns + o(n)$ bits is called ``compressed succinct
  index,'' whereas the redundancy is $o(n)$, which would fall under
  ``ultra-succinct'' for data structures.
  We should *re-think* this section: instead of cleaning up, we
  possibly become the source of more confusion!
\end{JOHANNES}

The terms of \emph{integrated encoding}~\cite{barbay07succinct2},
\emph{self-index}~\cite{MNspire07}, \emph{non-systematic
  data structure}~\cite{gal07cell,fischer10optimal} or \emph{encoding
  data structure}~\cite{brodal10space} refer to a data structure which
does not require access to any other data structure than itself, as opposed to a succinct index.
In the case of integrated encodings~\cite{barbay07succinct2} and
self-indices~\cite{MNspire07}, there is no need for any other data
structure, as they re-code all information and hence provide their own
mechanism for accessing the data.
\begin{VLONG}
  Those data structures are considered less practical from the point
  of view of modularity, but this approach has the advantage of
  yielding potentially lower redundancies:
  Golynski~\cite{golynski07optimal} showed that if a bit vector $B$ is
  stored verbatim using $n$ bits, then every index supporting the
  operators {\access}, {\rank}, and {\select} must have redundancy
  $\Omega(\frac{n\lg\lg n}{\lg n})$ bits, while
  P\v{a}tra\c{s}cu~\cite{patrascu08succincter} gave an integrated
  encoding for $B$ with redundancy $\Oh(\frac{ n }{\polylog n})$ bits.
\end{VLONG}
In the case of non-systematic
data structures~\cite{gal07cell,fischer10optimal} and encoding data
structures~\cite{brodal10space}, the emphasis is that those indexing
data structures require much less space than the data they index, and
being able to answer some queries (other than {\access}, obviously)
without any access to the main data. 
Of course, such an index can be seen as a data structure itself,
\emph{for a distinct data type} (e.g., a Lowest Common Ancestor 
non-systematic succinct index of $2n+o(n)$ bits for labeled trees is
also a simple data structure for ordinal trees): those notions are
relative to their context.

\begin{VLONG}
  Following the model of Daskalakis et al.'s
  analysis~\cite{2009-SODA-SortingAndSelectionInPosets-DaskalakisKarpMosselRuebsebfekdVerbin}
  of sorting algorithms for partial orders, we distinguish the
  \emph{data complexity} and the \emph{index complexity} of both
  algorithms and succinct indices, measuring separately the number of
  operations it performs on the data and on the index, respectively.
  Following these definitions, a non-systematic data structure is a
  succinct index of data complexity equal to zero, and the usual
  complexity of a succinct index is the sum of its data complexity
  with its index complexity.
  This distinction is important for instance when we consider a
  semi-external memory model, where it could occur that the data
  structure is too large to reside in main memory and is therefore
  kept in external memory (which is expensive to access), but its
  index is small enough to be stored in RAM.
  In such a case it is preferable to use a succinct index of minimal
  data complexity.
  \begin{INUTILE}
    In the case of algorithms, these definitions separate the amount
    of operations performed on the data (typically, comparisons in the
    comparison model) from the amount of operations performed on the
    meta-data gathered by the algorithm on the data itself (typically,
    multiplications and divisions in the comparison model, but also
    computation of a Huffman tree based on the lengths of the runs in
    the work from Barbay and Navarro described in
    Section~\ref{sec:adapt-sort-compr}).
    Such operations can have quite distinct cost (this is the case in
    particular in Daskalakis et al.'s
    study~\cite{2009-SODA-SortingAndSelectionInPosets-DaskalakisKarpMosselRuebsebfekdVerbin}),
    and it is easy to sum them to get the traditional complexity
    measure, if only for compression with previous results.
  \end{INUTILE}
\end{VLONG}

\begin{INUTILE}
  \begin{JEREMY}
    The results of this section were cited only once, in the proof of
    Theorem~\ref{th:nsruns}, so I replaced by a simple cite in the
    proof.
  \end{JEREMY}
  \subsection{Rank and Select on Binary Strings}
  \label{sect:rank}

  Consider a \emph{bit-string} $B[1,n]$ of length $n$. We define the
  fundamental \emph{rank}- and \emph{select}-operations on $B$ as
  follows: $\rank_1(B,i)$ gives the number of 1's in the prefix
  $B[1,i]$, and $\select_1(B,i)$ gives the position of the $i$-th 1 in
  $B$, reading $B$ from left to right ($1 \le i \le n$). Operations
  $\rank_0(B,i)$ and $\select_0(B,i)$ are defined analogously for
  0-bits.

  \begin{lemma}[Raman et al.~\cite{raman07succinct}]
    \label{lem:rank_and_select}
    Let $B$ be a bit-vector of length $n$ having $m$ 1's.
    There is a compressed data structure representing  $B$ 
    using $\lceil\lg{n\choose m}\rceil + o(n)$ bits and
    supporting $\rank_x(B,i)$ and $\select_x(B,i)$ in $\Oh(1)$ time
    for any $x \in \{0,1\}$ and $i\in[1..n]$.
    In particular, the space is $o(n)$ if $m=o(n)$.
  \end{lemma}

\end{INUTILE}

\subsection{\lrm~Trees }
\label{sect:lrm}

LRM-Trees are an elegant way to partition a sequence of
values into sorted consecutive blocks, and to express the relative
position of the first element of each block within a previous
block. 
They were introduced under this name as an internal tool for basic navigational
operations in ordinal
trees~\cite{sadakane10fully}
and, under the name of ``2d-Min Heaps,'' to index integer arrays
in order to support range minimum queries on
them~\cite{fischer10optimal}.

Let $A[1,n]$ be an integer array. 
For technical reasons, we define $A[0]=-\infty$ as the ``artificial''
overall minimum of the array.
 
\begin{definition}[Fischer~\cite{fischer10optimal}; Sadakane and Navarro~\cite{sadakane10fully}]
  \label{def:2dmin}
  For $1\le i \le n$, let $\psv_A(i)=\max\{j \in [0..i-1]~:~A[j] < A[i]\}$
  denote the \emph{previous smaller value} of position $i$.
  The \emph{\lrm~Tree (LRM-Tree)} $\mathcal{T}_A$ of $A$
  is an ordered labeled tree with vertices $0, \dots, n$.
  For $1 \le i \le n$, $\psv_A(i)$ is the parent node of $i$.  
  The children are ordered in increasing order from left to right.
\end{definition}

\begin{LONG}
  See Fig.~\ref{fig:minmin} for an example of LRM-Trees.
  \begin{figure}[t]
      \begin{center}
        \includegraphics[angle=-90,scale=1]{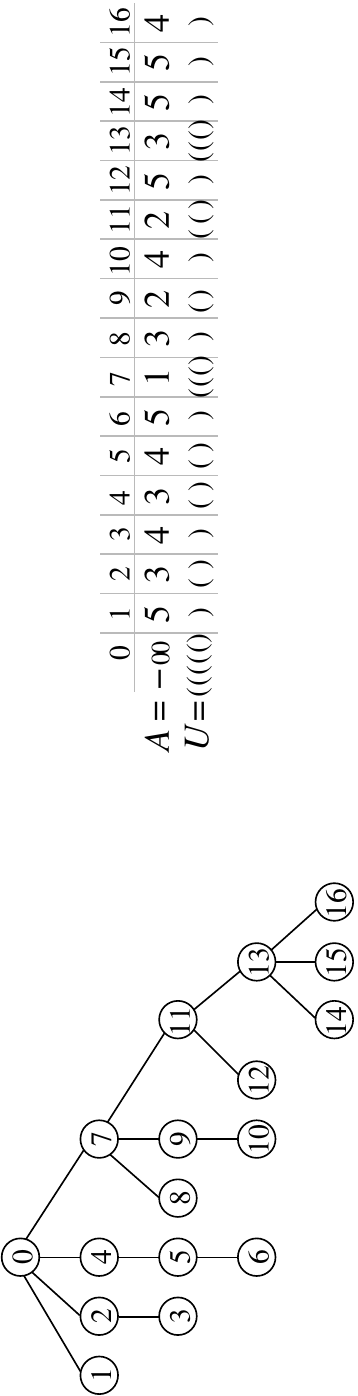}
      \end{center}
     \caption{The LRM-Tree $\mathcal{T}_A$ of the input array
      $A$, and its DFUDS $U$.}
    \label{fig:minmin}
  \end{figure}
\end{LONG}
The following lemma shows a simple way to construct the LRM-Tree in
linear time (Fischer \cite{fischer10optimal} gave a more complicated
linear-time algorithm with advantages that are irrelevant for this
paper.)

\begin{lemma}\label{lem:lrmConstruction}
  There is an algorithm computing the LRM-Tree of an array of $n$
  integers in at most $2n$ data comparisons.
\end{lemma}

\begin{proof}
  The computation of the LRM-Tree corresponds to a simple scan over
  the input array,
  starting at $A[0]=-\infty$, 
  building down iteratively the current rightmost branch of the tree
  with increasing elements of the sequence
  till an element $x$ smaller than its predecessor is
  encountered, 
  at which point one climbs the right-most branch up to the first node
  $v$ holding a value smaller than $x$, and starts a new branch with a
  right-most child of $v$ of value $x$.
  As the root of the tree has value $A[0]=-\infty$ smaller than all
  elements, the algorithm always terminates.
 
  The construction algorithm performs at most $2n$ comparisons.
  Charging the last comparison performed during the insertion of an
  element $x$ to $x$ itself, and all previous comparisons to the
  elements already in the LRM-Tree, each element is charged at most
  twice: once when it is inserted into the tree, and once when
  scanning it while searching for a smaller value on the rightmost
  branch.
  As in the latter case all scanned elements are removed from the
  rightmost path, this second charging occurs at most once for each
  element.
\end{proof}

\subsection{Range Minimum Queries}
\label{sec:minim-range-quer}

We consider the following queries on a static array $A[1,n]$ (parameters $i$ and $j$ with $1\le i\le j\le n$):

\begin{definition}[Range Minimum Queries]
  $\rmq_A(i,j) = $ position of the minimum in $A[i,j]$.
\end{definition}

RMQs have a wide range of applications for various data structures and
algorithms, including text indexing~\cite{fischer09faster}, pattern
matching~\cite{crochemore08improved}, and more elaborate kinds of
range queries~\cite{chen04range}.

The connection between LRM-Trees and RMQs is given
as follows. 
For two given nodes $i$ and $j$ in a tree $T$,
let $\lca_T(i,j)$ denote their \emph{Lowest Common Ancestor} (LCA),
which is the deepest node that is an ancestor of both $i$ and $j$.
Now let $\mathcal{T}_A$ be the LRM-Tree of $A$. 
For arbitrary nodes $i$ and $j$ in $\mathcal{T}_A$,
$1 \le i < j \le n$, let $\ell=\lca_{\mathcal{T}_A}(i,j)$.
Then if $\ell=i$, $\rmq_A(i,j)$ is given by $i$, and otherwise,
$\rmq_A(i,j)$ is given by the child of $\ell$ that is on the path
from $\ell$ to $j$ \cite{fischer10optimal}.

Since there are succinct data structures supporting the LCA
operator\footnote{The inherent connection between RMQs and LCAs has
  been exploited also in the other direction \cite{bender05lowest}.}
in succinctly encoded trees in constant time, this yields a succinct
index:
\begin{lemma}[Fischer\cite{fischer10optimal}]\label{lem:LRMTreeSpace}
  For an array $A[1,n]$ of totally ordered objects,
  there is a non-systematic succinct index
  using $2n+o(n)$ bits and
  supporting RMQs in zero data queries and $\Oh(1)$
  index queries. 
  This index can be built using at most $\Oh(n)$ data comparisons.
\end{lemma}
\begin{JEREMY}
  Two problems with this proof (hence I removed it): 
  1) if it was already given in your
  paper, we should not give it again, and if it is different we should
  explain why (and maybe not cite your paper in its title). 
  2) More importantly, it does not prove the linear time construction
  of the RMQ index. Does your paper prove it?
\end{JEREMY}
\begin{JOHANNES}
 ad 2:  YES OF COURSE I prove it!
\end{JOHANNES}
\begin{INUTILE}
  \begin{proof}
    We encode the LRM-Tree $\mathcal{T}_A$ by its \emph{Depth-First
      Unary Degree Sequence} (DFUDS) $U$ using $2n$
    bits~\cite{benoit05representing}.
    In $U$, nodes are listed in preorder, where node $i$ with $k$
    children is encoded as $k$ opening parentheses and one closing
    parenthesis (a single opening parenthesis is prepended to $U$ to
    make it balanced).
    As the nodes in $\mathcal{T}_A$ appear in preorder, we can jump to
    the description of node $i$ (corresponding to index $i$ in $A$) by
    a single select-statement in $U$.
    Further, there are $o(n)$-bit indices for DFUDS for computing LCAs
    in $\Oh(1)$ time~\cite{fischer10optimal}.
    By the discussion above, from the LCA we can compute the range
    minimum.
  \end{proof}
\end{INUTILE}

\subsection{Adaptive Sorting, and Compression of Permutations}
\label{sec:adapt-sort-compr}

Sorting a permutation of $n$ elements in the comparison model
typically requires $\Omega(n\lg n)$ comparisons in the worst case.
Yet, better results can be achieved for some parameterized classes of
permutations.
Among others,
Knuth~\cite{theArtOfComputerProgrammingVol3} considered \emph{Runs}
(ascending subsequences), counted by
$\nRuns(\pi)=1+|\{i~:\ 1 < i\leq n, \pi_{i+1}<\pi_i\}|;$
Levcopoulos and Petersson~\cite{sortingShuffledMonotoneSequences}
introduced \emph{Shuffled Up Sequences}, counted by
$\nSUS(\pi)=\min\{k:\pi\textrm{ is covered by }k\textrm{ increasing
  subsequences} \},$ and \emph{Shuffled Monotone Sequences}, counted
by $\nSMS(\pi) = \min\{k:\pi\textrm{ is covered by }k\textrm{ monotone
  subsequences} \};$
and Barbay and Navarro~\cite{barbay09compressed} introduced strict
variants of those concepts, namely \emph{Strict Runs} and \emph{Strict
  Shuffled Up Sequences}, where sorted subsequence are composed of
consecutive integers \begin{LONG}(e.g.
  $(\mathbf{2,3,4},1,\mathit{5,6,7,8})$ has two runs but three strict
  runs)\end{LONG}, counted by $\nsruns$ and $\nssus$, respectively.
For any ``measure of disorder'' $\nwhatever$ among those five, there
is a variant of the merge-sort algorithm which sorts a permutation
$\pi$ of size $n$ and measure $\nwhatever$ in time
$\Oh((n+1)\lg\nwhatever)$, which is optimal in the worst case among
instances of fixed size $n$ and fixed values of $\nwhatever$ (this is
not necessarily true for other measures of disorder).

As the merging cost induced by a subsequence is increasing with its
length, the sorting time of a permutation can be improved by
rebalancing the merging tree~\cite{barbay09compressed}.
\begin{LONG}
  This merging cost is actually equivalent to the cost of encoding,
  for each element of the sorted permutation, the subsequence of
  \emph{origin} of this element.
  Hence rebalancing the merging tree is equivalent to optimize a code
  for those addresses, and can be done via a Huffman
  tree~\cite{Huf52}.
\end{LONG}
The complexity can then be expressed more precisely in function of the
entropy of the relative sizes of the sorted subsequences identified,
where the \emph{entropy} $\entropy(\vPartition)$ of a sequence
$\vPartition=\langle n_1, n_2,\dots,n_r\rangle$ of $r$ positive
integers adding up to $n$ is
$\entropy(\vPartition)=\sum_{i=1}^r\frac{n_i}{n}\lg\frac{n}{n_i}$,
which satisfies $(r-1)\lg n \le n \entropy(\vPartition) \le n \lg r$
(by concavity of the logarithm).

\begin{JEREMY}
  Simplified the inequality (by multiplying each term by n), and
  realized it was different from the one Gonzalo used for STACS 2010:
  yours was saying $\frac{r\lg n}{n} \le \entropy(\vPartition) \le \lg
  r$ (by convexity of the logarithm)., which is equivalent to $r\lg n
  \le n \entropy(\vPartition) \le n \lg r$.  I did not have the time
  to check it and put the one from Gonzalo?

  Also, I corrected ``(by convexity of the logarithm)'': the logarithm
  is concave and not convex (check
  \url{http://en.wikipedia.org/wiki/Concave_function} if in doubt).
\end{JEREMY}

Barbay and Navarro~\cite{barbay09compressed} observed that each such
algorithm from the comparison model also describes an encoding of the
permutation $\pi$ that it sorts, so that it can be used to compress
permutations from specific classes to less than the information-theoretic
lower bound of $n\lg n$ bits.
Furthermore they used the similarity of the execution of the
merge-sort algorithm with a Wavelet Tree~\cite{grossi03high}, to
support the application of $\pi()$ and its inverse $\pi^{-1}()$ in
time logarithmic in the disorder of the permutation $\pi$ as measured
by $\nRuns$, $\nSRuns$, $\nSUS$, $\nSSUS$ and $\nSMS$, respectively) in
the worst case.
We summarize their technique in Lemma~\ref{lem:EntropySort} below, in
a way independent of the partition chosen for the permutation.

\begin{lemma}[Barbay et al.~\cite{barbay09compressed}]\label{lem:EntropySort}
  Given a partition $\vPartition$ of a permutation $\pi$ of $n$
  elements into $\nPartition$ sorted subsequences of respective
  lengths $\vPartition$,
  these subsequences can be merged with $n(1+\entropy(\vPartition))$
  comparisons on $\pi$ and $\Oh(\nPartition\lg\nPartition)$ internal
  operations, and this merging can be encoded
  using at most
  $(1+\entropy(\vPartition))(n+o(n))+\Oh(\nPartition\lg n)$
  bits so that it
  supports the computation of $\pi(i)$ and $\pi^{-1}(i)$ in time
  $\Oh(1+\lg\nPartition)$ in the worst case and in time
  $\Oh(1+\entropy(\vPartition))$ on average when $i$ is chosen
  uniformly at random in $[1..n]$.
\end{lemma}

\section{Compressed Succinct Indexes for Range Minima}

We now explain how to improve on the result from
Lemma~\ref{lem:LRMTreeSpace} for permutations that are partially
ordered.
Without loss of generality, we consider only the case where the input
is a permutation of $[1..n]$: if this is not the case, we can sort the
elements in $A$ by rank, considering earlier occurrences of equal
elements as smaller.

\subsection{Strict Runs}

The simplest compressed data structure for RMQs uses an amount of
space which is a function of $\nsruns$, the number of strict runs in
$\pi$. 
It uses $2\nsruns+o(n)$ bits on permutations where $\nsruns\in o(n)$:

\begin{theorem}
  \label{th:nsruns}
  There is a non-systematic  compressed succinct index
  using $2\nsruns+\lceil\lg{n\choose \nsruns}\rceil +o(n)$ bits and
  supporting RMQs in zero data queries and $\Oh(1)$ index queries.
  \begin{FORJOURNAL}
    Add Construction time, in the text and in the proof.
  \end{FORJOURNAL}
\end{theorem}
\begin{proof}
  We mark the beginnings of each runs in $A$ with a 1 in a bit-vector
  $B[1,n]$, and represent $B$ with the compressed succinct data
  structure from Raman et al.~\cite{raman07succinct}, using
  $\lceil\lg{n\choose \nsruns}\rceil + o(n)$ bits.
  Further, we define $A'$ as the (conceptual) array consisting of the
  heads of $A$'s runs ($A'[i]=A[\select_1(B,i)]$). 
  We build the LRM-Tree from Lemma~\ref{lem:LRMTreeSpace} on $A'$;
  using $2\nsruns(1+o(1))$ bits. 
  To answer a query $\rmq_A(i,j)$, compute $x=\rank_1(B,i)$ and
  $y=\rank_1(B,j)$, and compute $m'=\rmq_{A'}(x,y)$ as a range minimum
  in $A'$, and map it back to its position in $A$ by
  $m=\select_1(B,m')$. 
  Then if $m<i$, return $i$ as the final answer to $\rmq_A(i,j)$,
  otherwise return $m$.
  The correctness from this algorithm follows from the fact that only
  $i$ and the heads of strict runs that are \emph{entirely} contained
  in the query interval can be the range minimum; the former occurs if
  and only if the head of the run containing $i$ is smaller than all
  other heads in the query range.
\end{proof}

\begin{LONG}
  Obviously, this compressed data-structure is interesting only if
  $\nSRuns\in o(n)$. We explore in the following section a more
  general measure of partial order, $\nRuns$.
\end{LONG}

\subsection{General Runs}

The same idea as in Theorem~\ref{th:nsruns} applied to more general
runs yields another compressed succinct index for RMQs, potentially
smaller but this time requiring to access the input to answer RMQs.

\begin{theorem}
  \label{thm:nruns_systematic}
  There is a systematic compressed succinct  index 
  using $2\nruns+\lceil\lg{n\choose \nruns}\rceil +o(n)$ bits and
  supporting RMQs in $1$ data comparison and $\Oh(1)$ index
  operations.
  \begin{FORJOURNAL}
    Add Construction time, in the text and in the proof.
  \end{FORJOURNAL}
\end{theorem}
\begin{proof}
  We build the same data structures as in Theorem~\ref{th:nsruns}, now
  using $2\nruns+\lceil\lg{n\choose \nruns}\rceil +o(n)$ bits. 
  To answer a query $\rmq_A(i,j)$, compute $x=\rank_1(B,i)$ and
  $y=\rank_1(B,j)$. 
  If $x=y$, return $i$.
  Otherwise, compute $m'=\rmq_{A'}(x+1,y)$, and map it back to its
  position in $A$ by $m=\select_1(B,m')$.
  The final answer is $k$ if $A[k]<A[m]$, and $m$ otherwise.
\end{proof}

To achieve a non-systematic compressed succinct index whose space
usage is a function of $\nruns$, we need more space and a more heavy
machinery, as shown next.
The main idea is that a permutation with few runs results in a
compressible LRM-Tree, where many nodes have out-degree~1.

\begin{theorem}
  \label{thm:runs}
  There is a non-systematic compressed succinct index
  using $2\nruns\lg n+o(n)$ bits, and 
  supporting RMQs in zero data comparisons and $\Oh(1)$ index
  operations.
  \begin{FORJOURNAL}
    Add Construction time, in the text and in the proof.
  \end{FORJOURNAL}
\end{theorem}
\begin{proof}
  We build the LRM-Tree $\mathcal{T}_A$ from Sect.~\ref{sect:lrm}
  directly on $A$, and then compress it with the tree-compressor due to
  Jansson et al.~\cite{jansson07ultra}. 

  To see that this results in the claimed space, let $n_k$ denote the
  number of nodes in $\mathcal{T}_A$ with out-degree $k\ge 0$.
  Let $(i_1,j_1),\dots,(i_\nruns,j_\nruns])$ be an encoding of the
  runs in $A$ as (start, end), and
  look at a pair $(i_x,j_x)$. 
  We have $\psv_A(k)=k-1$ for all $k\in[i_x+1..j_x]$, and so the nodes
  in $[i_x..j_x]$ form a path in $\mathcal{T}_A$, possibly interrupted
  by branches stemming from heads $i_y$ of other runs $y>x$ with
  $\psv_A(i_y) \in [i_x..j_x-1]$. 
  Hence $n_0=\nruns$, and $n_1\ge n-\nruns-(\nruns-1) \ge n-2\nruns$,
  as in the worst case the values $\psv_A(i_y)$ for
  $i_y\in\{i_2,i_3,\dots,i_\nruns\}$ are all different.

  Now $\mathcal{T}_A$, with degree-distribution $n_0,\dots,n_{n-1}$,
  is compressed into $nH^*(\mathcal{T}_A) + O\left(\frac{n\lg^2n}{\lg
      n}\right)$ bits, where
  $$
  nH^*(\mathcal{T}_A) = \lg \left(\frac{1}{n} \binom{n}{n_0 n_1 \dots
      n_{n-1}}\right)
  $$
  is the so-called \emph{tree entropy}~\cite{jansson07ultra} of
  $\mathcal{T}_A$.
  This representation supports all navigational operations in
  $\mathcal{T}_A$ in constant time, and in particular those required
  for Lemma~\ref{lem:LRMTreeSpace}.
  A rough inequality yields a bound on the number of possible
  LRM-Trees:
  $$      
  \binom{n}{n_0 n_1 \dots n_{n-1}}  
  =    \frac{n!}{n_0! n_1! \dots n_{n-1}!}
  \le  \frac{n!}{n_1!}
  \le  \frac{n!}{(n-2\nruns)!}
  \le  n^{2\nruns}\ ,
  $$
  \begin{INUTILE}
    \begin{eqnarray*}
      \binom{n}{n_0 n_1 \dots n_{n-1}} & = & \frac{n!}{n_0! n_1! \dots n_{n-1}!}\\
      & \le & \frac{n!}{n_1!}\\
      & \le & \frac{n!}{(n-2\nruns)!}\\
      & \le & n^{2\nruns}\ , \\ 
    \end{eqnarray*}
  \end{INUTILE}
  from which one easily bounds the space usage of the compressed
  succinct index:
  $$
  nH^*(T)  
  \le  \lg\left(\frac{1}{n} n^{2\nruns}\right)
  =  \lg\left(n^{2\nruns-1}\right)
  =  (2\nruns-1) \lg n
  \le  2\nruns \lg n\ . 
  $$
  \begin{INUTILE}
    \begin{eqnarray*}
      nH^*(T) & \le & \lg\left(\frac{1}{n} n^{2\nruns}\right)\\
      & = & \lg\left(n^{2\nruns-1}\right)\\
      & = & (2\nruns-1) \lg n\\
      & \le & 2\nruns \lg n\ . 
    \end{eqnarray*}
  \end{INUTILE}
  Adding the space required to index the structure of Jansson et
  al.~\cite{jansson07ultra}  yields the desired space.
\end{proof}

\begin{INUTILE}
  \subsection{Range Quantile Queries}
  We now turn our attention to range quantile queries.

  \begin{theorem}
    \label{thm:runs}
    There is a non-systematic preprocessing scheme for
    $\Oh(\lg\nssus)$-RQQs that needs $n(2+H(\ssus))(1+o(1))$ bits. If
    $p$ is chosen uniformly at random in $[1..n]$ then the average time
    to answer $\rqq_A(\cdot,\cdot,p)$ is $\Oh(1 + H(\ssus))$.
  \end{theorem}
  \begin{proof}
    Our basis is the approach of Gfeller and Sanders
    \cite{gfeller09towards}, also described by Gagie et
    al.~\cite{gagie09range}, that represents $A$ by a binary wavelet
    tree \cite{grossi03high}. The difference is that the shape of the
    wavelet tree is not perfectly balanced, but determined by the tree
    resulting from running the Hu-Tucker algorithm \cite{hu71optimum}
    on $\ssus$. In particular, let
    $[i_1..j_1],\dots,[i_\nssus..j_\nssus]$ be the shuffled up-sequences
    with $i_0=1, i_x=j_{x-1}+1$ for $2\le x\le\nssus$, and
    $j_\nssus=n$. Then $\ssus = \langle y_1,\dots,y_\nssus \rangle$
    with $y_x=j_x-i_x+1$, and the Hu-Tucker algorithm produces a
    binary tree $T$ where the $k$'th leftmost leaf $\ell_k$ represents
    the interval $[i_k..j_k]$, and $L=\sum y_k\textit{depth}(\ell_k)$
    is minimal. Moreover, $L < n(2 + H(\ssus))$.

    We enhance $T$ with the usual structures of a wavelet tree: every
    node $v$ (implicitly) stores an subsequence $A_v[1,n_v]$ of the
    original array $A$ (the root stores the complete array
    $A[1,n]$). An associated bit-vector $B_v[1,n_v]$ indicates which
    characters from $A_v$ can be found in the left or right child of
    $v$: $B_v[i]=0$ iff $A_v[i] \in [i_k..j_k]$, and the $k$'th
    leftmost leaf is in the subtree rooted at $v$'s left
    child. Bet-vectors $B_v$ also define the sequences of their
    respective children in a natural way. This process continues
    recursively until we have reached a leaf of $T$, where no
    information is stored.

    Now a query $\rqq_A(i,j,p)$ proceeds as in the original algorithm
    \cite{gfeller09towards}, working its way down the tree while
    adapting $i$, $j$, and $k$ accordingly. Now suppose we have
    reached a leaf $\ell$ with values $i'$, $j'$, and $p'$. There, we
    search for the position of the $p'$'th smallest element in
    $A_\ell$. But $A_\ell=[i_k..j_k]$ for some $k$, so the answer is
    simply $i'+p'-1$.

    Because the depth of $T$ is at most $\nssus$, the claim on
    worst-case time follows. The claim on average time follows from
    the fact that the average depth of $T$ is given by $H(\ssus)$.
  \end{proof}

\end{INUTILE}

\section{Sorting Permutations}
\label{sec:sort-perm}

Barbay and Navarro~\cite{barbay09compressed} showed how to use the
decomposition of a permutation $\pi$ in $\nRuns$ ascending consecutive
runs of respective lengths $\vRuns$ to sort adaptively to their
entropy $\entropy(\vRuns)$.
Those runs entirely partition the LRM-Tree of $\pi$ into $\nRuns$
paths, each starting at some branching node of the tree, and ending at
a leaf: one can easily draw this partition by iteratively tagging the
leftmost maximal untagged up-from-leaf path of the LRM-Tree.

Yet, any partition of the LRM-Tree into down paths \begin{LONG}(so
  that the values traversed by the path are increasing)\end{LONG} can
be used to sort $\pi$.
Since there are exactly $\nRuns$ leaves in the LRM-Tree, no such
partition can be smaller than the partition of $\pi$ into ascending
consecutive runs. 
But in the case where some of those partitions are more imbalanced
than the original one, this yields a partition of smaller entropy, and
hence a faster sorting algorithm.
We define a family of such partitions:
\begin{definition}[LRM-Partition]\label{def:LRMTreePartitioning}
  A \emph{LRM-Partition} of a permutation $\pi$ with LRM-Tree
  $\mathcal{T}_\pi$ is defined recursively as follows.
  One subsequence is the ``spinal chord'' of
  $\mathcal{T}_\pi$, one of the longest root-to-leaf paths in
  $\mathcal{T}_\pi$.
  Removing this spinal chord of $\mathcal{T}_\pi$ leaves a forest of
  more shallow trees.
  The rest of the partition is obtained by computing and concatenating
  some LRM-partitions of those trees.
\end{definition}
\begin{LONG}
  This definition does not define a unique partition, but a family of
  partitions: there might be several ways to choose the ``spinal
  chord'' of each subtree when several nodes have the same depth, and
  of course the order of the subsequences in the partition does not
  matter either.
  Yet, there will always be $\nruns$ many subsequences in the
  partition, and any LRM-Partition is never worse and often better (in
  terms of sorting \emph{and} compressing) than the the original
  Run-Partition.
  The situation is similar to the one of $\entropy(\vSUS)$ versus
  $\nSUS$: it is easier to minimize $\nSUS$ (resp. $\nRuns$) than
  $\entropy(\vSUS)$ (resp. $\entropy(\vLRM)$), yet one can take
  advantage of the entropy of a partition minimizing $\nSUS$ (resp. of
  a LRM-Partition).
\end{LONG}

Note that each down-path of the LRM-Tree corresponds to an ascending
subsequence of $\pi$, but not all ascending subsequences correspond to
down-paths of the LRM-Tree, hence partitioning optimally $\pi$ into
$\nSUS$ ascending subsequences potentially yields smaller partitions,
or ones of smaller entropy: the LRM-partitions seem inferior to
SUS-partitions.
Yet, the fact which make LRM-Partitions particularly interesting is
that it can be computed in linear time (which is not true for
SUS-Partitions):

\begin{lemma}\label{lem:LRMTreePartitioning}
  There is an algorithm finding one of the LRM-Partitions of a
  permutation $\pi$ of size $n$ in $\Oh(n)$ data comparisons.
\end{lemma}
\begin{proof}
  Definition~\ref{def:LRMTreePartitioning} is constructive: we are
  only left to show that this algorithm can be executed in linear
  time.
  Having built $\mathcal{T}_A$ using Lemma \ref{lem:lrmConstruction}
  in $2n$ comparisons, we first set up an array $D$ containing the
  \emph{depths} of the nodes in $\mathcal{T}_A$, listed in
  preorder. 
  We then index $D$ for range maximum queries in linear time using
  Lemma \ref{lem:LRMTreeSpace}.

  Now the deepest node in $\mathcal{T}_A$ can be found by a range
  maximum query over the whole array, supported in constant time.
  From this node, we follow the path to the root, and save the
  corresponding nodes as the first subsequence. 
  This divides $A$ into disconnected subsequences, which can be
  processed recursively using the same algorithm, as the nodes in any
  sub-tree of $\mathcal{T}_A$ form an interval in $D$. 
  We do so until all elements in $A$ have been assigned to a
  subsequence. 

  Note that in the recursive steps, the numbers in $D$ are not anymore
  the depths of the corresponding nodes in the remaining
  sub-trees. 
  But as all depths listed in $D$ differ by the same offset from their
  depths in any connected subtree, this does not affect the result of the range maximum query.
\end{proof}

Given a LRM-Partition of the permutation $\pi$, sorting $\pi$ is just
a matter of applying Lemma~\ref{lem:EntropySort}:

\begin{theorem}\label{th:LRMTreeSorting}
  Let $\pi$ be a permutation of size $n$.
  Identifying its $\nRuns$ runs by building the LRM-Tree through
  Lemma~\ref{lem:lrmConstruction},
  obtaining a LRM-Partition of subsequences of respective lengths
  $\vLRM$ through Lemma~\ref{lem:LRMTreePartitioning},
  and merging the subsequences of this partition through
  Lemma~\ref{lem:EntropySort},
  results in an algorithm sorting $\pi$ in a total of
  $n(3+\entropy(\vLRM))$ data comparisons and 
  $\Oh(n+\nRuns\lg\nRuns)$ internal operations,
  accounting for a total time of 
  $\Oh(n(1+\entropy(\vLRM)))$.
\end{theorem}

\begin{proof}
  Lemma~\ref{lem:lrmConstruction} builds the LRM-Tree in $2n$ data
  comparisons,
  Lemma~\ref{lem:LRMTreePartitioning} extract from it a LRM-Partition
  in $\Oh(n)$ internal operations, and
  Lemma~\ref{lem:EntropySort} merges the subsequences of the
  LRM-Partition in $n(1+\entropy(\vLRM))$ data comparisons and
  $\Oh(\nRuns\lg\nRuns)$ internal operations.
  The sum of those complexities yields $n(3+\entropy(\vLRM))$ data
  comparisons and $\Oh(n+\nRuns\lg\nRuns)$ internal operations.

  Since $\nRuns\lg\nRuns < n\entropy(\vLRM)+\lg\nRuns$ by concavity of
  the logarithm, the total time complexity is in
  $\Oh(n(1+\entropy(\vLRM)))$.
\end{proof}

Since by construction $\entropy(\vLRM)\leq\entropy(\vRuns)$, this
result naturally improves on the adaptive merge sort algorithm for runs \cite{barbay09compressed}.
However, $\entropy(\vSUS)$ can be arbitrarily smaller than
$\entropy(\vLRM)$: this means that, in the worst case over instances
of fixed $n$ and $\entropy(\vSUS)$, SUS sorting has a strictly better
asymptotical complexity than LRM sorting; while, in the worst case
over instances of fixed $n$ and $\entropy(\vLRM)$, SUS sorting has the
same asymptotical complexity than LRM sorting.
\begin{JEREMY}
  ``instances of fixed $n$ and $\entropy(\vLRM)$'' is not well
  defined, since $\entropy(\vLRM)$ is not uniquely defined...
\end{JEREMY}

Yet, on instances where $\entropy(\vLRM)<2\entropy(\vSSUS)-1$,
LRM-Sorting actually performs \emph{less} data comparisons (and
potentially more index operations) than SUS-Sorting.
Barbay~et al.'s improvement~\cite{barbay09compressed} of SUS-Sorting
performs $2n(1+\entropy(\vSUS)$ data comparisons, decomposed into
$n(1+\entropy(\vSUS))$ data comparisons to compute a partition $\pi$
into $\nSUS$ sub-sequences which is minimal in size, if not
necessarily in entropy; and
$n(1+\entropy(\vSUS))$ data comparisons (and $\Oh(n+\nSUS\lg\nSUS)$
internal operations) to merge the subsequences into a single ordered
one.
On the other hand, the combination of Lemma~\ref{lem:lrmConstruction}
with Lemma~\ref{lem:LRMTreePartitioning} yields a LRM-Partition in
$2n$ data comparisons and $\Oh(n)$ index operations;
which is then merged in $n(1+\entropy(\vLRM))$ data comparisons (and
$\Oh(n+\nRuns\lg\nRuns)$ internal operations) to merge the
subsequences into a single ordered one.
Comparing the $2n(1+\entropy(\vSUS)$ data comparisons of SUS-Sorting
with the $n(3+\entropy(\vLRM))$ data comparisons of LRM-Sorting shows
that on instance where $\entropy(\vLRM)<2\entropy(\vSUS)-1$,
LRM-Sorting performs less data comparisons (only potentially twice
less, given that $\entropy(\vSUS)\leq\entropy(\vLRM)$.
This comes to the price of potentially more internal operations:
SUS-Sorting performs $\Oh(n+\nSUS\lg\nSUS)$ such ones while
LRM-Sorting performs $\Oh(n+\nRuns\lg\nRuns)$ such ones, and
$\nSUS\leq\nRuns$ by definition.

When considering external memory, this is important in the case where
the data does not fit in main memory while the internal
data-structures (using much less space than the data itself) of the
algorithms do: then data comparisons are much more costly than
internal operations.
Furthermore, we show in the next section that this difference of
performance implies an even more meaningful difference in the size of
the permutation encodings corresponding to the sorting algorithms.

\section{Compressing Permutations}
\label{sec:compr-perm}

As shown by Barbay and Navarro \cite{barbay09compressed}, sorting opportunistically in the
comparison model yields a compression scheme for permutations, and
sometimes a compressed succinct data structure supporting the direct
and inverse operators in reasonable time.
We show that this time again the sorting algorithm of
Theorem~\ref{th:LRMTreeSorting} corresponds to a compressed succinct
data structure for permutations which supports the direct and reverse
operators in good time, while often using less space than previous
solutions.
The essential component of our solution is a data structure encoding
the LRM-Partition.
In order to apply Lemma~\ref{lem:EntropySort}, our data structure
must support two operators in good time:
\begin{itemize}
\item the first operator, $\map(i)$, consists of indicating, for each
  position $i\in[1..n]$ in the input permutation $\pi$, the
  corresponding subsequence $s$ of the LRM-Partition, and the relative
  position $p$ of $i$ in this subsequence;
\item the second operator, $\unmap(s,p)$ is just the reverse of the
  previous one: given a subsequence $s\in[1..\nRuns]$ of the
  LRM-Partition of $\pi$ and a position $p\in[1..n_s]$ in it, the
  operator must indicate the corresponding position $i$ in $\pi$.
\end{itemize}

We obviously cannot afford to rewrite the numbers of $\pi$ in the
order described by the partition, which would use $n\lg n$ bits.
\begin{LONG}
  A naive solution would be to encode this partition as a string $S$
  over alphabet $[1..\nRuns]$, using a succinct data structure
  supporting the {\access}, {\rank} and {\select} operators on it.
  This solution is not suitable as it would require at the very least
  $n\entropy(\vRuns)$ bits
  \emph{only to encode the LRM-Partition}, making this encoding worse
  than the $\nRuns$ compressed succinct data structure \cite{barbay09compressed}.
\end{LONG}
We describe a more complex data structure which uses linear space, and
supports the desired operators in constant time.
\begin{lemma}\label{lem:LRMPartitionDataStructure}
  Let $P$ be a LRM-Partition consisting of $\nRuns$ subsequences of
  respective lengths $\vLRM$, summing to $n$.
  There is a succinct data structure 
  using $2(n+\nRuns)+o(n)$ bits and
  supporting the operators $\map$ and $\unmap$ on $P$ in constant
  time.
\end{lemma}
\begin{proof}
  The main idea of the data structure is that the subsequences of a
  LRM-Partition for a permutation $\pi$ are not as general as, say,
  the subsequences of the partition into $\nSUS$ up-sequences.
  For each pair of subsequences $(u,v)$, either the positions of $u$
  and $v$ belongs to distinct intervals of $\pi$, or the values
  corresponding to $u$ (resp. $v$) all fall between two values from
  $v$ (resp. $u$).

  As such, the subsequences of the LRM-Partition can be organized into
  a forest of ordinal trees,
  where the internal nodes of the trees correspond to the $\nRuns$
  subsequences of the LRM-Partition, organized so that $u$ is parent
  of $v$ if the positions of $v$ are contained between two positions
  of $u$, and
  where the leaves of the trees correspond to the $n$ positions in
  $\pi$, children of the internal node $u$ corresponding to the
  subsequence they belong to.
  \begin{LONG}
    For instance, the permutation $\pi=(4,5,9,6,8,1,3,7,2)$ has a
    unique LRM-Partition $\{(4,5,6,8),(9),(1,3,7),(2)\}$, whose
    encoding can be visualized by the expression $(45(9)68)(137)(2)$
    and encoded by the balanced parenthesis expression
    $(()()(())()())(()()())(())$ (note that this is a forest, not a
    tree, hence the excess of '('s versus ')'s is going to zero several times inside the
    expression).
  \end{LONG}

  Given a position $i\in[1..n]$ in $\pi$,
  the corresponding subsequence $s$ of the LRM-Partition is simply
  obtained by finding the parent of the $i$-th leaf, and returning its
  preorder rank among internal nodes.
  The relative position $p$ of $i$ in this subsequence is given by the
  number of its left siblings which are leaves.
  Given a subsequence $s\in[1..\nRuns]$ of the LRM-Partition of $\pi$
  and a position $p\in[1..n_s]$ in it, the corresponding position $i$
  in $\pi$ is computed by finding the $s$-th internal node in
  preorder, selecting its $p$-th child which is a leaf, and computing
  the preorder rank of this node among all the leaves of the tree.

  We represent such a forest using a Balanced Parentheses
  Sequence
  using $2(n+\nRuns)+o(n)$ bits and enhance it with a
  $o(n)$-bit succinct index \cite{sadakane07compressed}
  supporting in constant time the operators
  {\rank} and {\select} on leaves (i.e., on the pattern
  '()'), and
  {\rank} and {\select} on internal nodes (i.e., on the pattern
  '((').
  With these operators we can simulate all operations described in the
  previous paragraph.
\end{proof}

Given the data structure for LRM-Partitions from
Lemma~\ref{lem:LRMPartitionDataStructure}, applying the merging
data structure from Lemma~\ref{lem:EntropySort} immediately yields a
compressed succinct data structure for permutations.
Note that this encoding is \emph{not} a succinct index, so that it
would not make any sense to measure its space complexity in term of
data and index complexity.

\begin{theorem}\label{th:LRMPermutationDataStructure}
  Let $\pi$ be a permutation of size $n$ and $P$ a LRM-Partition for
  $\pi$ consisting of $\nRuns$ subsequences of respective lengths
  $\vLRM$.
  There is a compressed succinct data structure
  using $(1+\entropy(\vLRM))(n+o(n))+\Oh(\nRuns\lg n)$ bits,
  supporting the computation of $\pi(i)$ and $\pi^{-1}(i)$ in time
  $\Oh(1+\lg\nRuns)$ in the worst case, and in time
  $\Oh(1+\entropy(\vLRM))$ on average when $i$ is chosen uniformly at
  random in $[1..n]$, and
  which can be computed in the times indicated in
  Theorem~\ref{th:LRMTreeSorting}, summing to
  $\Oh(n(1+\entropy(\vLRM)))$.
\end{theorem}
\begin{proof}
  Lemma~\ref{lem:LRMPartitionDataStructure} yields a data structure
  for a LRM-Partition of $\pi$ using $2(n+\nRuns)+o(n)$ bits, and
  supports the $\map$ and $\unmap$ operators in constant time.
  The merging data structure from Lemma~\ref{lem:EntropySort} requires
  $(1+\entropy(\vLRM))(n+o(n))+\Oh(\nRuns\lg n)$ bits, and
  supports the operators $\pi()$ and $\pi^{-1}()$ in the time
  described, through the additional calls to $\map$ and $\unmap$.
  Summing both spaces yields the desired final space.
\end{proof}
\begin{JOHANNES}
  Jeremy, pls carefully check the above.
\end{JOHANNES}
\begin{JEREMY}
  I did, I think.  One confusion I have: Computing the Hufman takes
  only $\nPartition\lg\nPartition$ comparisons, but encoding the lengths of the
  runs requires something like $\nPartition\lg n$ bits. This is smaller than
  $n$ bits anyway so we could hide it away.
\end{JEREMY}

\section{Conclusion and Future Work}
\label{sec:concl-future-work}

%
One additional result not described here is
\begin{LONG}
  how to take advantage of strict runs, in addition of
  taking advantage of general runs, for LRM sorting and encoding of
  permutation.
  Another related result is
\end{LONG}
a variant of LRM-Trees, \emph{Roller Coaster Trees} (RC-Trees), which
take advantage of permutations formed by the combinations of ascending
and descending runs.
\begin{LONG}
  This approach is trivial when considering subsequences of
  consecutive positions, gets slightly technical when considering the
  insertion of descending runs, and requires new techniques to adapt
  the compressed succinct data structure to this new setting.
  Since the optimal partitioning into up and down sequences when
  considering general subsequences requires exponential time,
  RC-Sorting seems a much desirable improvement on merging ascending
  and descending runs, as well as a more practical alternative to
  SMS-Sorting, in the same way as LRM-Tree improved on Runs-Sorting
  while staying more practical than SUS-Sorting.
\end{LONG}
Another result to come is the generalization of our results to the
indexing, sorting and compression of general sequences (i.e., also to
integer functions),
\begin{JOHANNES}
  But we do support RMQs for general sequences, don't we?
  (It's what we claim, and we discussed it a while ago.)
\end{JOHANNES}
taking advantage of the redundancy in a general
sequence to sort faster and encode in even less space, in function of
both the entropy of the frequencies of the symbols and the entropy of
the lengths of the subsequences of the LRM-Partition.
Finally, studying the integration of those compressed data structures
into compressed text indexes like suffix arrays \cite{navarro07compressed}
is likely to yield interesting results, too.

\bibliographystyle{abbrv}
\bibliography{paper}

\begin{INUTILE}
\newpage\appendix

  \section{Proof of Lemma~\ref{lem:EntropySort}}

\end{INUTILE}
\end{document}